\documentclass[9pt]{extarticle}
\usepackage{spconf}
\usepackage{amsthm}
\usepackage{amsfonts}
\usepackage{amssymb}
\usepackage{mathrsfs}
\usepackage{mathtools}
\usepackage[english]{babel}
\usepackage{float}
\usepackage[toc,page]{appendix}
\usepackage[pdftex]{graphicx}
\usepackage{epstopdf}
\usepackage{amsmath}
\usepackage{color}
\usepackage{multirow}
\usepackage{graphicx}
\usepackage{bm}
\usepackage{balance}
\usepackage{microtype}
\usepackage[caption=false,font=footnotesize]{subfig}

\usepackage{url}

\usepackage{cite}
\usepackage[]{algorithmicx}
\usepackage{algpseudocode,algorithm} 

\usepackage{color}
\usepackage[table]{xcolor}% http://ctan.org/pkg/xcolor
\usepackage{lipsum}

\DeclareMathOperator*{\argmax}{argmax} % thin space, limits underneath in displays
\newtheorem{theorem}{Theorem}

\title{NBA-OMP: Near-field Beam-Split-Aware Orthogonal Matching Pursuit for Wideband THz Channel Estimation}
\name{Ahmet~M.~Elbir$^{\dagger}$, Kumar Vijay Mishra$^{+,\dagger}$ and   Symeon Chatzinotas$^\dagger$%\thanks{Thanks to XYZ agency for funding.}
}
\address{	$^{\dagger}$Interdisciplinary Centre for Security, Reliability and Trust, University of Luxembourg, Luxembourg \\
	$^+$United States DEVCOM Army Research Laboratory, Adelphi, USA \\
	\thanks{	This work was supported in part by  the ERC project AGNOSTIC and Horizon project TERRAMETA.}\\
}

\begin{document}
	\maketitle
	
	\begin{abstract}
		The sixth-generation networks envision the terahertz (THz) band as one of the key enabling technologies because of its ultrawide bandwidth. To combat the severe attenuation, the THz wireless systems employ large arrays, wherein the near-field beam-split (NB) severely degrades the accuracy of channel acquisition. Contrary to prior works that examine only either narrowband beamforming or far-field models, we estimate the wideband THz channel via an NB-aware orthogonal matching pursuit (NBA-OMP) approach. We design an NBA dictionary of near-field steering vectors by exploiting the corresponding angular and range deviation. Our OMP algorithm accounts for this deviation thereby \textit{ipso facto} mitigating the effect of NB. Numerical experiments demonstrate the effectiveness of the proposed channel estimation technique for wideband THz systems.
	\end{abstract}


	\section{Introduction}
	\label{sec:Introduciton}
	In recent years, there has been considerable research interest in Terahertz (THz) band for sixth-generation  (6G) wireless cellular networks because of the availability of extremely wide bandwidth in THz spectrum. In particular, THz wireless communications is envisioned to provide significant enhancements in data rate ($>100\text{Gb/s}$), extremely low propagation latency ($<1\text{ms}$), and ultra-reliability ($99.999\%$)~\cite{thz_Rappaport2019Jun}. To this end, several signal processing challenges in this band are yet to be addressed including severe path loss arising from fading and molecular absorption, extremely-sparse channel model, very short transmission ranges, and beam-split; see \cite{elbir2022Aug_THz_ISAC,ummimoTareqOverview} and references therein, for details. In general, the severe path loss is addressed by deploying ultra-massive MIMO  architectures, wherein subcarrier-independent (SI) analog beamformers are employed. This is analogous to massive multiple-input multiple-output (MIMO) arrays used in millimeter-wave (mm-Wave) systems~\cite{heath2016overview,elbir2022Nov_Beamforming_SPM}. In a wideband ultra-massive MIMO array, the directions of the generated beams at different subcarriers may point to different directions, or \textit{beam-split}, because the analog beamformers are designed with respect to a single subcarrier~\cite{thz_beamSplit,elbir2022Aug_THz_ISAC}. 
	
	Further, conventional wireless systems operating at sub-$6$ GHz and mm-Wave bands employ far-field plane-wave models. But the THz communications range is much shorter than far-field assumption may not always hold~\cite{ummimoTareqOverview}. In particular, when the transmission range is shorter than the Fraunhofer distance, the wavefront is spherical in the near-field \cite{nf_primer_Bjornson2021Oct}. As a result, channel acquisition algorithms should take into account both direction and range information %(see, e.g., Fig.~\ref{fig_BS}) for accurate signal processing~
	\cite{elbir2022Aug_THz_ISAC}. 
	
	The existing THz channel estimation techniques mostly rely on far-field signal model~\cite{dovelos_THz_CE_channelEstThz2,elbir2022Jul_THz_CE_FL,elbir_THZ_CE_ArrayPerturbation_Elbir2022Aug,thz_channelEst_beamsplitPatternDetection_L_Dai}. The near-field scenarios in \cite{nf_OMP_Dai_Wei2021Nov,nf_mmwave_CE_noBeamSplit_Cui2022Jan} are limited to mm-Wave and ignore the effect of beam-split. The narrowband processing suggested in \cite{nf_OMP_Dai_Wei2021Nov} considers hybrid (near- and far-field) models. On the other hand, several methods have been proposed to compensate the far-field beam-split for both THz channel estimation~\cite{elbir2022Jul_THz_CE_FL,elbir_THZ_CE_ArrayPerturbation_Elbir2022Aug,dovelos_THz_CE_channelEstThz2,thz_channelEst_beamsplitPatternDetection_L_Dai} and beamforming~\cite{elbir2021JointRadarComm,elbir2022_thz_beamforming_Unified_Elbir2022Sep,delayPhasePrecoding_THz_Dai2022Mar} applications.  Nevertheless, THz channel estimation in the presence of near-field beam-split (NB) remains relatively unexamined.  
	
	In this work, we introduce an NB-aware orthogonal matching pursuit (NBA-OMP) approach for wideband THz channel estimation. First, we define the NB model in terms of both user direction-of-arrival (DoA) and range parameters. Then, we design an NBA dictionary whose columns are subcarrier-dependent (SD), which spans the whole angular spectrum and transmission range up to the Fraunhofer distance.  While the degree of beam-split is proportionally known prior to the DoA/range estimation, it depends on the unknown user location. For example, consider $f_m$ and $f_c$ to be the frequencies for the, respectively, $m$-th and center subcarriers. When $\theta$ is the physical user direction, then the spatial direction corresponding to the $m$-th-subcarrier is shifted by $\frac{f_c}{f_m}\theta$. We then employ the OMP algorithm, which accounts for this deviation, thereby \textit{ipso facto} compensating the effect of NB. Numerical experiments demonstrate the effectiveness of the proposed approach that we compare with the existing THz channel-estimation-based methods~\cite{nf_OMP_Dai_Wei2021Nov,thz_channelEst_beamsplitPatternDetection_L_Dai}.
	
	The range-dependent beampattern is observed in certain far-field applications such as frequency diverse array (FDA) radars \cite{lv2022co}, which employ linear frequency offsets across the antennas to yield a range dependent beampattern. This enables a joint estimation \cite{lv2022co} of target angle and range, including in the presence of clutter \cite{lv2022clutter}. However, the FDA wavefront is not spherical. Rydberg synthetic apertures \cite{vouras2023overview,vouras2023phase} also exhibit similar complex beampatterns but do not employ large arrays. Contrary to these prior works, we focus on spherical wavefront near-field for a THz communications system that employs an extremely large array and transmits multiple subcarriers via orthogonal frequency-division multiplexing (OFDM) signaling.

	%	\lipsum[1-5]

	The rest of the paper is organized as follows. In the next section, we introduce the signal model for a multi-user wideband THz ultra-massive MIMO system and describe the NB effect in Section~\ref{sec:NFBP}. Section~\ref{sec:NBAOMP} presents the proposed NBA-OMP followed by numerical experiments in Section~\ref{sec:Sim}. We conclude in Section~\ref{sec:Conc}.	%\textit{Notation:} 
	Throughout this paper,  $(\cdot)^\textsf{T}$ and $(\cdot)^{\textsf{H}}$ denote the transpose and conjugate transpose operations, respectively. $[\mathbf{a}]_k$  correspond to the   $k$-th entry of a vector $\mathbf{a}$, while $\mathbf{A}^{\dagger}$ denotes the Moore-Penrose pseudo-inverse of matrix $\mathbf{A}$. A unit matrix of size $N$ is represented by $\mathbf{I}_N$. $\Sigma(a) = \frac{\sin N\pi a }{N\sin \pi a}$ is the Dirichlet sinc function.
	
	%\newpage
	\section{System Model}
	\label{sec:probForm}
	\vspace{-8pt}
	Consider a wideband (OFDM) THz MIMO architecture with hybrid analog/digital beamforming over $M$ subcarriers. We assume that the base station (BS) has $N$ antennas and $N_\mathrm{RF}$ radio-frequency (RF) chains to serve $K$ single-antenna users. Define the data symbols  $\mathbf{s}[m] = [s_1[m],\cdots,s_K[m]]^\textsf{T}$  for $m\in \mathcal{M} = \{1,\cdots, M\}$, which are processed via a $K\times K$ SD baseband beamformer $\mathbf{F}_\mathrm{BB}[m] = [\mathbf{f}_{\mathrm{BB}_1}[m],\cdots,\mathbf{f}_{\mathrm{BB}_K}[m]]$. To steer the generated beams toward users in the downlink, an  $N\times N_\mathrm{RF}$ SI analog beamformer $\mathbf{F}_\mathrm{RF}$ ($N_\mathrm{RF}=K<{N}$) is employed. Since the analog beamformers are realized via phase-shifters, they have a constant-modulus constraint, i.e. $|[\mathbf{F}_\mathrm{RF}]_{i,j}| = \frac{1}{\sqrt{N}}$ as $i = 1,\cdots, N_\mathrm{RF}$ and $j = 1,\cdots, N$. Then, the transmitted signal, i.e.,  $\mathbf{F}_\mathrm{RF}\mathbf{F}_\mathrm{BB}[m]\mathbf{s}[m]$, is received by the $k$-th user at the $m$-th subcarrier as 
	\begin{align}
	\label{receivedSignal}
	{y}_{k}[m] = \mathbf{h}_{k}^\textsf{T}[m]\sum_{i = 1}^{K}\mathbf{F}_\mathrm{RF}\mathbf{f}_{\mathrm{BB}_i}[m]{s}_i[m]  + {w}_k[m],
	\end{align}
	where ${w}_k[m]\in \mathbb{C}$ is the complex additive white Gaussian noise (AWGN) vector with ${w}_k[m] \sim \mathcal{CN}({0},\sigma_n^2)$.

	%			\subsection{Near-field Channel Model}
	At THz, the reflected path components and scattering is not significant. The channel is usually modeled as the superposition of a single LoS path with a few assisting NLoS paths~\cite{ummimoTareqOverview,elbir2021JointRadarComm,thz_beamSplit}.	In addition, multipath channel models are also widely used, especially for indoor applications~\cite{teraMIMO,ummimoTareqOverview}. Hence, we consider a general scenario, wherein the $N\times 1$ channel matrix for the $k$-th user at the $m$-th subcarrier is represented by the combination of $L$ paths as~\cite{ummimoTareqOverview}
	\begin{align}
	\label{channelModel}
	\mathbf{h}_k[m]  =  
	\sqrt{\frac{N}{L}}  \sum_{l =1}^{L}   \alpha_{k,m,l} \mathbf{a}(\phi_{k,l},r_{k,l})   e^{-j2\pi\tau_{k,l} f_m },
	\end{align}
	where  $\tau_{k,l}$ represents the time delay of the $l$-th path corresponding to the array origin.  $\alpha_{k,m,l}\in\mathbb{C}$  denotes the complex path gain and the expected value of its magnitude for the indoor THz multipath model is $	\mathbb{E}\{|\alpha_{k,m,l}|^2 \} = \left(\frac{c_0}{4\pi f_m r_{k,l} } \right)^2 e^{- k_\mathrm{abs}(f_m) r_{k,l}  },$	where $c_0$ is the speed of light, $f_m$ is the $m$-th subcarrier frequency, $r_{k,l}$ represents the distance from the $k$-th user to the array origin, and  $k_\mathrm{abs}(f_m)$ is the SD medium absorption coefficient~\cite{ummimoTareqOverview,ummimoHBThzSVModel,thz_clusterBased_Yuan2022Mar}. Furthermore $f_m = f_c + \frac{B}{M}(m - 1 - \frac{M-1}{2}) $, where  $f_c$ and $B$ are carrier frequency and bandwidth, respectively.

	The high-frequency operation at THz combined with an extremely small array aperture implies that the close-proximity users are in near-field, where planar wave propagation is not valid. At ranges shorter than the Fraunhofer distance $F = \frac{2 D^2}{\lambda}$, where $D$ is the array aperture and $\lambda = \frac{c_0}{f_c}$ is the wavelength, the near-field wavefront is spherical \cite{nf_primer_Bjornson2021Oct,elbir_THZ_CE_ArrayPerturbation_Elbir2022Aug}.  For a uniform linear array (ULA), the array aperture is $D = (N-1)d$, where $d = \frac{\lambda}{2}$ is the element spacing. At THz, near-field signal model should be employed because $r_{k,l} <F$. For instance, when $f_c = 300$ GHz and $N=256$, the Fraunhofer distance is $F = 32.76$ m.

	Taking into account the spherical-wave model~\cite{nf_primer_Bjornson2021Oct,nf_Fresnel_Cui2022Nov}, we define the near-field steering vector $\mathbf{a}(\phi_{k,l},r_{k,l})\in\mathbb{C}^{N}$ corresponding to the physical DoA  $\phi_{k,l}$ and range $r_{k,l}$ as 
	\begin{align}
	\label{steeringVec1}
	\mathbf{a}(\phi_{k,l},r_{k,l}) = \frac{1}{\sqrt{N}} [e^{- \mathrm{j}2\pi \frac{d}{\lambda}r_{k,l}^{(1)} },\cdots,e^{- \mathrm{j}2\pi \frac{d}{\lambda_m}r_{k,l}^{(N)} }]^\textsf{T},
	\end{align}
	where $\phi_{k,l} = \sin \tilde{\phi}_{k,l}$ with $\tilde{\phi}_{k,l}\in [-\frac{\pi}{2},\frac{\pi}{2}]$, and  $r_{k,l}^{(n)}$ is the distance between the $k$-th user and the $n$-th antenna as
	\begin{align}
	r_{k,l}^{(n)} = \left(r_{k,l}^2  + 2(n-1)^2 d^2 - 2 r_{k,l}(n-1) d \phi_{k,l}   \right)^{\frac{1}{2}}. \label{eq:rkln}
	\end{align}
	Following the Fresnel approximation ~\cite{nf_Fresnel_Cui2022Nov,nf_primer_Bjornson2021Oct}, \eqref{eq:rkln} becomes
	\begin{align}
	\label{r_approx}
	r_{k,l}^{(n)} \approx r_{k,l}  - (n-1) d \phi_{k,l}  + (n-1)^2 d^2 \zeta_{k,l}  ,
	\end{align}	 
	where $\zeta_{k,l} = \frac{1- \phi_{k,l}^2}{2 r_{k,l}}$. Rewrite (\ref{steeringVec1}) as
	\begin{align}
	\label{steeringVectorPhy}
	\mathbf{a}(\phi_{k,l},r_{k,l}) \approx e^{- \mathrm{j}2\pi \frac{f_c}{c_0}r_{k,l}} \tilde{\mathbf{a}}(\phi_{k,l},r_{k,l}),
	\end{align} where the $n$-th element of $\tilde{\mathbf{a}}(\phi_{k,l},r_{k,l})\in \mathbb{C}^N$ is 
	\begin{align}
	\label{steeringVectorPhy2}
	[\tilde{\mathbf{a}}(\phi_{k,l},r_{k,l})]_n = e^{\mathrm{j} 2\pi \frac{f_c}{c_0}\left( (n-1)d\phi_{k,l}  - (n-1)^2 d^2 \zeta_{k,l}\right) }.
	\end{align}
	The steering vector in (\ref{steeringVectorPhy}) corresponds to the physical location $(\phi_{k,l},r_{k,l})$. This deviates to the spatial location $(\bar{\phi}_{k,m,l},\bar{r}_{k,m,l})$ in the beamspace because of the absence of SD analog beamformers. Then, the $n$-th entry of the deviated steering vector in (\ref{steeringVectorPhy2}) for the spatial location is 
	\begin{align}
	\label{steeringVectorSpa}
	&[\tilde{\mathbf{a}}(\bar{\phi}_{k,m,l},\bar{r}_{k,m,l})]_n \hspace{-3pt}= \hspace{-2pt}e^{\mathrm{j} 2\pi \frac{f_m}{c_0}\left( (n-1)d\bar{\phi}_{k,m,l}  - (n-1)^2 d^2 \bar{\zeta}_{k,m,l}\right) }.
	\end{align}

	Compared to mm-Wave frequencies, the THz bandwidths are so wide that a single-wavelength assumption for beamforming cannot hold. This leads to the \textit{split} of physical DoA/ranges $\{\phi_{k,l},r_{k,l}\}$ in the spatial domain. In order to estimate the THz channel accurately, our goal is to recover the physical DoA/range parameters from the deviated spatial spectrum via proposed signal processing techniques. %In what follows, we first introduce the near-field beam-split, then present our proposed approach. 
	
	%In what follows, we first define the NB, then introduce our NBA-OMP approach for near-field channel estimation. 
	%\textcolor{red}{Our goal is to ....???}
	
	\section{Near-Field Beam-Split}
	\label{sec:NFBP}
	\vspace{-8pt}
	%\subsection{NB model}
	%In order to introduce the NB model, we 
	We introduce the following Theorem 1 to establish the relationship between the physical and spatial DoAs/ranges. 
	\begin{theorem}
		Denote $\mathbf{u}\in \mathbb{C}^N $ and $\mathbf{v}_m \in \mathbb{C}^N$ as the arbitrary near-field steering vectors corresponding to the physical (i.e., $\{\phi_{k,l},r_{k,l}\}$) and spatial (i.e., $\{\bar{\phi}_{k,m,l},	\bar{r}_{k,m,l}\}$) locations given in (\ref{steeringVectorPhy2}) and (\ref{steeringVectorSpa}), respectively. Then, in spatial domain at subcarrier frequency $f_m$, the array gain achieved by $\mathbf{u}^\textsf{H}\mathbf{v}_m$ is maximized and the generated beam is focused at the location $\{\bar{\phi}_{k,m,l},	\bar{r}_{k,m,l}\}$ such that  
		\begin{align}
		\label{physical_spatial_directions}
		\bar{\phi}_{k,m,l} =    \eta_m \phi_{k,l}, \hspace{5pt}
		\bar{r}_{k,m,l} =    \frac{1 - \eta_m^2 \phi_{k,l}^2}{\eta_m(1 -\phi_{k,l}^2)}r_{k,l},
		\end{align}
		where  	 $\eta_m = \frac{f_c}{f_m}$ represents the proportional deviation of DoA/ranges.
	\end{theorem}
	
	\begin{proof}
		
		%%-----------------------------------------------------
		\begin{figure*}[t!]
			\centering		{\includegraphics[draft=false,width=0.85\textwidth]{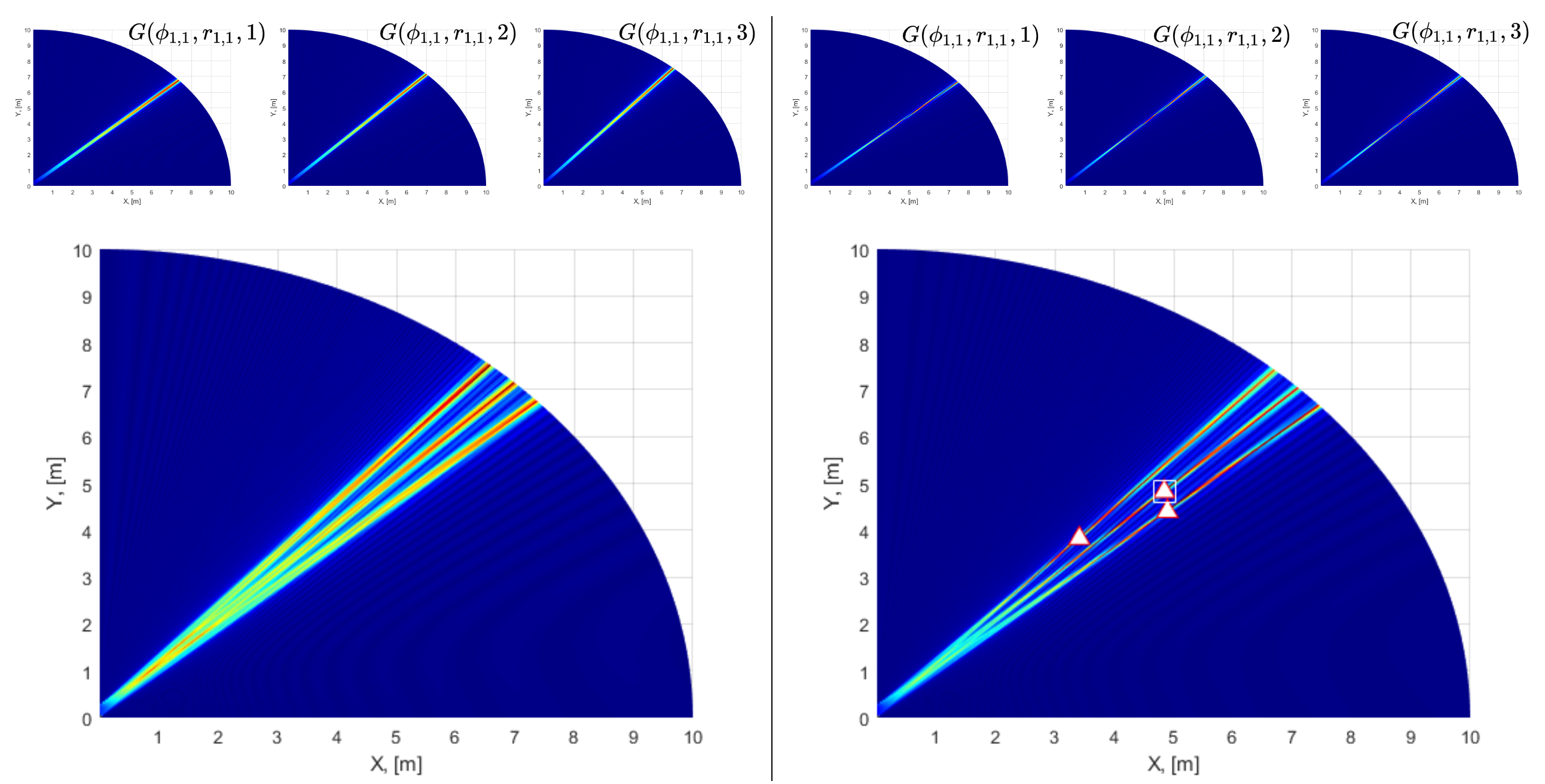}} 
			\caption{Array gains $G(\phi_{1,1},r_{1,1},{m})$ in Cartesian coordinates for a single user ($K=1$, $L=1$) located in the far-field $(45^\circ,6000\text{m})$ (left) and near-field $(45^\circ,6\text{m})$ (right), respectively. Here, $M=3$, $f_c=300$ GHz, and $B=30$ GHz. The top panel shows the gain for different subcarriers which are summed up to produce a composite array gain at the bottom for both far- and near-field cases clearly showing the beam-split. The square represents the user location while the triangles correspond to the spatial locations (where the maximum array gain is achieved) at different subcarriers. Whereas the far-field beam-split is only angular, the near-field split is across both range and angular domains.%\textcolor{red}{in (a), what do the three lines indicate? %Also, if this is array gain, why is it flat at the extremities? Are you plotting only the main lobe? 
				%hat is denoted by the blue background in a limited region?} 
				\vspace{-8pt}}
			%			\vspace*{-5mm}
			\label{fig_BS}
		\end{figure*}
		%%-----------------------------------------------------

		%\textcolor{red}{The combined array gain () for $k=1$ $L=1$ }
		
		Define the array gain achieved by $\mathbf{v}_m$ on an arbitrary user location $\{\phi_{k,l},r_{k,l}\}$ with steering vector $\mathbf{u}$ as
		\begin{align}
		&G(\phi_{k,l},r_{k,l},{m}) = \frac{|\mathbf{u}^\textsf{H} \mathbf{v}_m|^2}{N^2} \nonumber\\
		& = \frac{1}{N^2} \left|\sum_{n = 0}^{N-1} e^{\mathrm{j}\frac{2\pi}{c_0}  \left[ nd\left(f_m\bar{\phi}_{k,m,l} - f_c \phi_{k,l} \right) - n^2d^2(f_m\bar{\zeta}_{k,m,l} - f_c \zeta_{k,l} )  \right]    }   \right|^2, \nonumber \\
		&=  \frac{1}{N^2} \left|\sum_{n = 0}^{N-1} e^{\mathrm{j}\frac{2\pi n}{c_0}  (f_m\kappa_m -  f_c\kappa )    }   \right|^2, \label{arrayGain1}
		\end{align}
		where $\kappa = d({\phi}_{k,l} - n d \zeta_{k,l}) $ and $\kappa_m = d(\bar{\phi}_{k,m,l} - n d \bar{\zeta}_{k,m,l}) $,  for which $\bar{\zeta}_{k,m,l} =  \eta_m \zeta_{k,l}  = \frac{1 - \bar{\phi}_{k,m,l}^2}{2 \bar{r}_{k,m,l}}$. Then, (\ref{arrayGain1}) becomes
		\begin{align}
		G(\phi_{k,l},r_{k,l},{m})&= \frac{1}{N^2} \left| \frac{1 - e^{-\mathrm{j}2\pi N (f_m\kappa_m -  f_c\kappa)   }}{1 - e^{-\mathrm{j}2\pi (f_m\kappa_m -  f_c\kappa)  }}   \right|^2  \nonumber\\
		&= \frac{1}{N^2}\left| \frac{\sin (\pi N(f_m\kappa_m -  f_c\kappa) )}{\sin (\pi (f_m\kappa_m -  f_c\kappa))}    \right|^2 \nonumber \\
		& = |\Sigma( f_m\kappa_m -  f_c\kappa )|^2, \label{arrayGain}
		\end{align}
		where  the Dirichlet sinc function $\Sigma(\cdot)$  implies that most of the power is focused only on a small portion of the beamspace because of the power-focusing in $\Sigma(\cdot)$. This power substantially reduces across the subcarriers as $|f_m - f_c|$ increases. Furthermore, maximum array gain is achieved when $|\Sigma( a)|^2$ peaks when $a = 0$, i.e., $f_m\bar{\phi}_{k,m,l} =  f_c\phi_{k,l}$ and $f_m\bar{\zeta}_{k,m,l} = f_c\zeta_{k,l}$.  Therefore, we have $\bar{\phi}_{k,m,l} = \eta_m \phi_{k,l} $. Using $	\bar{\zeta}_{k,m,l} =  \eta_m \zeta_{k,l}$, we get $\bar{r}_{k,m,l} =   \frac{1 - \eta_m^2 \phi_{k,l}^2}{\eta_m(1 -\phi_{k,l}^2)}r_{k,l}$.
	\end{proof}

	Following (\ref{r_approx}) and (\ref{physical_spatial_directions}),  define NB in terms of DoAs and ranges of the users  as, respectively,
	\begin{align}
	\label{beamSplit2}
	\Delta(\phi_{k,l},m) &= \bar{\phi}_{k,m,l} - \phi_{k,l} = (\eta_m -1)\phi_{k,l}, 
	\end{align}
	and $\Delta(r_{k,l},m) = \bar{r}_{k,m,l} - r_{k,l} = (\eta_m -1)r_{k,l}$, i.e., 
	\begin{align}
	\Delta(r_{k,l},m) =  (\eta_m -1) \frac{1 - \eta_m^2 \phi_{k,l}^2}{\eta_m(1 -\phi_{k,l}^2)}r_{k,l}.
	\end{align}
	In Fig.~\ref{fig_BS}, the array gain is computed for both far- and near-field cases, wherein the coordinates with maximum array gain are achieved at different locations for different subcarriers because of beam-split.

	\section{NBA-OMP-Based Channel Estimation}
	\label{sec:NBAOMP}
	\vspace{-8pt}
	%\subsection{NBA dictionary design}
	The key idea of the proposed NBA dictionary design is to utilize the prior knowledge of $\eta_m$ to obtain beam-split-corrected steering vectors. In other words, for an arbitrary physical DoA and range, we readily obtain the spatial DoA and ranges as $\bar{\phi} = \eta_m\phi$ and $\bar{r} = \frac{1 - \eta_m^2 \phi^2}{\eta_m(1 - \phi^2)} r$. We exploit this observation to design the BSA dictionary $\mathcal{C}_m$ composed of steering vectors $\mathbf{c}(\phi_m,r_m)\in\mathbb{C}^N$ as
	\begin{align}
	\mathcal{C}_m = \{\mathbf{c}(\phi_m,r_m) |  \phi_m\in [-\eta_m,-\eta_m], r_m \in \mathbb{R}^+  \},
	\end{align}
	where the $n$-th element of $\mathbf{c}(\phi_m,r_m)$ is 
	\begin{align}
	\label{steeringVec2c}
	[\mathbf{c}(\phi_m,r_m)]_n =  e^{\mathrm{j} 2\pi \frac{f_m}{c_0}\left(  (n-1)d\phi_m  - (n-1)^2 d^2 \zeta_m\right) },
	\end{align}
	where $\zeta_m= \frac{1 - \phi_m^2}{2r_m}  $.

	Using the BSA dictionary $\mathcal{C}_m$, instead of SI steering vectors $\mathbf{a}(\phi,r)$, the SD virtual steering vectors $\mathbf{c}(\phi_m,r_m)$ is constructed for the OMP algorithm. Once the beamspace spectra is computed via OMP, the sparse channel support corresponding to the SD spatial DoA and ranges are obtained. This readily yields the physical DoAs and ranges as $\phi = \phi_m/\eta_m$ and ${r} =    \frac{\eta_m(1 - \phi^2)}{1 - \eta_m^2 \phi^2}r_m$,  $\forall m\in \mathcal{M}$. 	
	%$[\mathbf{a}(\psi)]_n = e^{j2\pi (n-1) \frac{d}{\lambda_M} \psi }$ and 
	%	\begin{align}
	%	\boldsymbol{\Gamma}_m(\theta) = \mathrm{diag}\{e^{j \eta_m \xi_1 },\cdots, e^{j \eta_m \xi_N }  \}\in\mathbb{C}^{N\times N},
	%	\end{align}
	%	where $\xi_n$ is the $n$th element of $N\times 1$ vector $\boldsymbol{\xi}$ which is computed as $ \boldsymbol{\xi}= \mathrm{unwrap}\{[\angle[\mathbf{a}(\theta)]_1,\cdots, \angle[\mathbf{a}(\theta)]_N]^\textsf{T}  \}$. 
	The proposed NBA dictionary $\mathcal{C}_m$ also holds spatial orthogonality as  $\lim_{N \rightarrow +\infty} |\mathbf{c}^\textsf{H}(\phi_{m,i},r_{m,i}) \mathbf{c}(\phi_{m,j},r_{m,j})  | = 0, \forall  i\neq j$. In the next section, we present the channel estimation procedure with the proposed NBA dictionary.

	%		\lipsum[1-20]

	%\subsection{Near-field channel estimation}
	In downlink, the channel estimation stage is performed simultaneously  during channel training by all the users indexed by $k\in \mathcal{K} = \{1,\cdots, K\}$. Since the BS employs hybrid beamforming architecture, it activates only a single RF chain in each channel use to transmit the pilot signals during channel acquisition~\cite{heath2016overview}. Hence, the BS employs $P$ beamformer vectors as $\tilde{\mathbf{F}}= [\tilde{\mathbf{f}}_1,\cdots, \tilde{\mathbf{f}}_P]\in \mathbb{C}^{N\times P}$ ($|\tilde{\mathbf{f}}_p| = 1/\sqrt{N}$) to send $P$ orthogonal pilots $\tilde{\mathbf{S}}[m] = \mathrm{diag}\{\tilde{s}_1[m],\cdots, \tilde{s}_P[m]\}\in \mathbb{C}^{P\times P}$, which are collected by the $k$-th user as
	\begin{align}
	\mathbf{y}_k[m] = \tilde{\mathbf{S}}[m]\bar{\mathbf{F}}[m] \mathbf{h}_k[m] + \mathbf{w}_k[m],
	\end{align}
	where $\bar{\mathbf{F}}[m] = \tilde{\mathbf{F}}^\textsf{T}[m]\in \mathbb{C}^{P\times N} $. Assume that $\bar{\mathbf{F}}[m] = \mathbf{F}\in \mathbb{C}^{P\times N}$ and $\tilde{\mathbf{S}}[m]= \mathbf{I}_P,\forall m\in \mathcal{M}$~\cite{limitedFeedback_Alkhateeb2015Jul,dovelos_THz_CE_channelEstThz2,thz_channelEst_beamsplitPatternDetection_L_Dai}, we get
	\begin{align}
	\label{y_vector}
	\mathbf{y}_k[m] = \mathbf{F}\mathbf{h}_k[m] + \mathbf{w}_k[m].
	\end{align}
	In the sparse reconstruction framework, (\ref{y_vector}) becomes
	\begin{align}
	\mathbf{y}_k[m] = \mathbf{F} \mathbf{C}_m \mathbf{x}_k[m] + \mathbf{w}_k[m],
	\end{align}
	where $\mathbf{C}_m = [\mathbf{c}(\phi_{m,1},r_{m,1}),\cdots,\mathbf{c}(\phi_{m,Q},r_{m,Q}) ]$
	is the ${N\times Q}$ NBA dictionary matrix covering the spatial domain with $\phi_{m,q} \in [-\eta_m,\eta_m]$ ($\phi \in [-1,1]$) and $r_{m,q}\in [0, F]$ for $q=1,\cdots, Q$. $\mathbf{x}_k[m]\in\mathbb{C}^{Q}$ is an $L$-sparse vector, whose non-zero elements correspond to the set $\{ x_{k,l}[m]| x_{k,l}[m]\triangleq\sqrt{\frac{N}{L}} \alpha_{k,m,l}e^{-j2\pi \tau_{k,l}f_m} \}$. 
	
	We employ the NBA dictionary matrix $\mathbf{C}_m$ for OMP-based sparse recovery. Algorithm~\ref{alg:BSACE} summarizes the proposed NBA-OMP procedure, wherein the physical DoAs and ranges are obtained by combining the beam-split-corrected DoAs/ranges for $m\in\mathcal{M}$ (Step $4-6$). Then, the estimated THz channel $\hat{\mathbf{h}}_k[m]$ is constructed from the set of steering vectors of physical DoAs/ranges (Step $13-16$).

	%-------------------------------------------------------------------------------------------------
	\begin{algorithm}[H]
		\begin{algorithmic}[1] 
			\caption{ \bf NBA-OMP}
			\Statex {\textbf{Input:}  Dictionary $\mathbf{C}_m$, observation $\mathbf{y}_k[m]$ and $\eta_m$ \label{alg:BSACE}}
			
			\State \textbf{for} $k \in\mathcal{K}$
			\State \indent $l=1$, $\mathcal{I}_{l-1} = \emptyset$,
			%			, $\boldsymbol{\Xi}_k = \mathbf{0}_{N\times L }$. 
			$\mathbf{r}_{l-1}[m] = \mathbf{y}_k[m], \forall m\in \mathcal{M}$.
			\State \indent\textbf{while} $l\leq L$ \textbf{do}
			
			\State \indent\indent $q^\star = \argmax_q \sum_{m=1}^{M}|\mathbf{c}^\textsf{H}(\phi_{m,q},r_{m,q})\mathbf{F}^\textsf{H}\mathbf{r}_{l-1}[m] |$.
			\State \indent\indent $\mathcal{I}_{l} \gets \mathcal{I}_{l-1} \bigcup \{q^\star\}$.
			\State \indent \indent  $\hat{\phi}_{k,l} =\frac{\phi_{m,q^\star}}{\eta_m}$, $\hat{r}_{k,l} =    \frac{\eta_m(1 - \phi_{m,q^\star}^2)}{1 - \eta_m^2 \phi_{m,q^\star} } r_{m,q^\star}$.
			\State \indent \indent $\hat{\Delta}(\phi_{k,l},m) = (\eta_m-1) \hat{\phi}_{k,l}$, $\forall m\in \mathcal{M}$.
			\State \indent \indent $\hat{\Delta}(r_{k,l},m) =(\eta_m -1) \frac{1 - \eta_m^2 \hat{\phi}_{k,l}^2}{\eta_m(1 -\hat{\phi}_{k,l}^2)}\hat{r}_{k,l} $, $\forall m\in \mathcal{M}$.
			\State \indent \indent $\boldsymbol{\Psi}_m(\mathcal{I}_l) \gets \mathbf{F}\mathbf{C}_m(\mathcal{I}_1)$.
			\State \indent \indent $\mathbf{r}_{l}[m] \gets \left( \mathbf{I}_P -  \boldsymbol{\Psi}_m(\mathcal{I}_l) \boldsymbol{\Psi}_m^\dagger(\mathcal{I}_l) \right) \mathbf{y}_k[m]$.%, $\forall m\in \mathcal{M}$.
			\State \indent\indent $l \gets l + 1$.
			\State \indent\textbf{end while}
			\State \indent  $\boldsymbol{\Xi}_k = [\mathbf{a}(\hat{\phi}_{k,1},\hat{r}_{k,1}), \cdots, \mathbf{a}(\hat{\phi}_{k,L}, \hat{r}_{k,L})]$.
			\State \indent \textbf{for} $m\in \mathcal{M}$
			\State \indent\indent $\hat{\mathbf{h}}_k[m] = \boldsymbol{\Xi}_k\hat{\mathbf{u}}_k[m]$, $\hat{\mathbf{u}}_k[m] = \boldsymbol{\Psi}_m^\dagger(\mathcal{I}_{l-1}) \mathbf{y}_k[m]$.
			\State \indent\textbf{end for}
			\State \textbf{end for}
			\Statex \textbf{Return:}   $\hat{\mathbf{h}}_k[m]$, $\hat{\Delta}(\phi_{k,l},m)$ and $\hat{\Delta}(r_{k,l},m)$.
		\end{algorithmic} 
	\end{algorithm}
	%------------------------------------------------------------------------------------------------
	
	\textit{Complexity and Training Overhead:} The computational complexity of NBA-OMP is same as the traditional OMP~\cite{heath2016overview}. It is primarily because of the matrix multiplications in step $4$ $(O(QMNP(NP + P)))$, step $9$ $(O(PNL\bar{L}))$, step $10$ $(O(P^2L\bar{L} + P^2   ))$ and step $15$  $(O(L(P + N)))$, where $\bar{L}=(L+1)/2$. Hence, the overall complexity is $O( QMNP(NP + P) + (PL\bar{L} + 1)(N + P) )$. 	The channel training overhead of the NBA-OMP  requires only $P$ ($8$ times lower, see Section~\ref{sec:Sim}) channel usage for pilot signaling than the traditional approaches (least-squares (LS) and minimum mean-squared-error (MMSE) estimation) that need at least $N$ times channel usage.
	
	\vspace{-8pt}
	\section{Numerical Experiments}
	\label{sec:Sim}
	\vspace{-8pt}
	We evaluate the performance of our NBA-OMP approach, in comparison with the state-of-the-art channel estimation techniques, e.g., far-field OMP (FF-OMP)~\cite{beamSquintRodriguezFernandez2018Dec}, near-field OMP (NF-OMP)~\cite{nf_OMP_Dai_Wei2021Nov}, beam-split pattern detection (BSPD)~\cite{thz_channelEst_beamsplitPatternDetection_L_Dai} as well as LS and MMSE. Throughout the experiments, unless stated otherwise, the signal model is constructed with $f_c=300$ GHz, $B=30$ GHz, $M=128$, $K=N_\mathrm{RF} = 8$, $L=3$, $P=8$ and $N=256$, for which $F = 32.76$ m. The NBA dictionary matrix is constructed with $Q=10N$, and the user directions and ranges are selected as $\phi_{k,l}\in \mathrm{unif}[-1,1]$, $r_{k,l} \in \mathrm{unif}[5,30]$ m, respectively. 
	
	Fig.~\ref{fig_NMSE_SNR} shows the channel estimation normalized MSE (NMSE) performance with respect to the signal-to-noise-ratio (SNR). The proposed NBA-OMP outperforms the competing methods while closely following the MMSE. The superior performance of NBA-OMP is attributed to accurate compensation of beam-split in both direction and range parameters via the NBA dictionary. %, which is composed of SD steering vectors. 
	On the other hand, the remaining methods fail to exhibit low error in the high SNR regime. % because theyeither fail to take into account the NB~\cite{nf_OMP_Dai_Wei2021Nov} or only consider the far-field signal model~\cite{alkhateeb2016frequencySelective,thz_channelEst_beamsplitPatternDetection_L_Dai}. 
	Furthermore, since the OMP is performed in the digital domain, the proposed approach does not require additional hardware, e.g., time-delayers as in~\cite{dovelos_THz_CE_channelEstThz2,delayPhasePrecoding_THz_Dai2022Mar}, to realize the SD dictionary matrices steering vectors. %, hence it is hardware-efficient.
	Fig.~\ref{fig_NMSE_BW} compares the NMSE performance with respect to the bandwidth $B\in [0,100]$ GHz. Our NBA-OMP effectively compensate the impact of beam-split for a large portion of the bandwidth up to $B<70$ GHz. 
	%				This is because accurate physical DoA/range parameters are obtained from the NBA dictionary while the remaining methods including BSPD, which involves beam-split-correction stage while near-field model is ignored. 

	%%-----------------------------------------------------
	\begin{figure}[t]
		\centering		{\includegraphics[draft=false,width=.85\columnwidth]{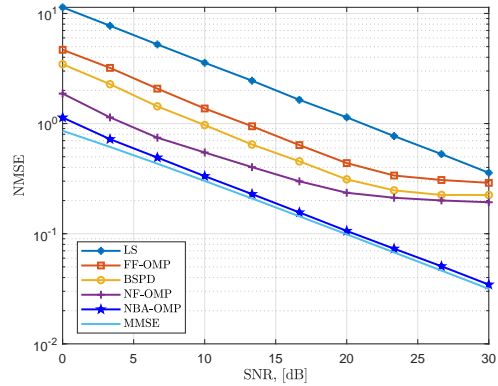}}  
		\caption{Near-field THz wideband channel estimation NMSE versus SNR. $N=256$, $f_c=300$ GHz, $M=128$ and $B=30$ GHz. \vspace{-8pt}
		}
		%			\vspace*{-5mm}
		\label{fig_NMSE_SNR}
	\end{figure}
	%%-----------------------------------------------------

	%%-----------------------------------------------------
	\begin{figure}[t]
		\centering		{\includegraphics[draft=false,width=.85\columnwidth]{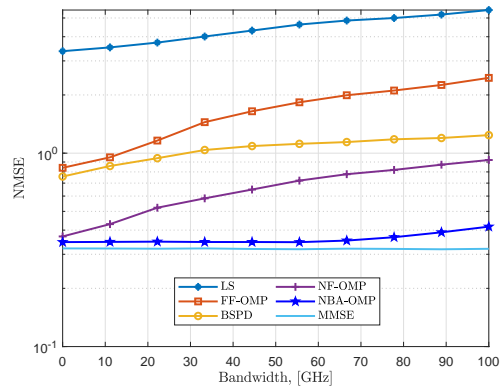}}  
		\caption{Near-field THz wideband channel estimation NMSE versus bandwidth when  $\mathrm{SNR}=10 $ dB.\vspace{-8pt}
		}
		%			\vspace*{-5mm}
		\label{fig_NMSE_BW}
	\end{figure}
	%%-----------------------------------------------------

	\vspace{-10pt}
	\section{Summary}
	\label{sec:Conc}
	\vspace{-8pt}
	We introduced an NBA-OMP approach to offset the impact of beam-split for near-field channel estimation. The proposed approach is especially beneficial at THz, wherein the signal wavefront becomes spherical for the close-proximity users and the far-field assumption is no longer valid. We evaluated the performance of our NBA-OMP channel estimation %in comparison with state-of-the-art techniques 
	and demonstrated that it is very close to the MMSE performance.

	%\newpage
	
	%	\footnotesize
	\balance
	\bibliographystyle{IEEEtran}
	\bibliography{IEEEabrv,references_118}

\end{document}